\documentclass{article}
\usepackage{lscape}
\usepackage{arxiv}
\usepackage[utf8]{inputenc}
\usepackage{amssymb}
\usepackage{todonotes}
\usepackage{amsmath}
\usepackage{amsthm}
\usepackage{amsfonts}
\usepackage{subfigure}
\usepackage{tikz}
\usetikzlibrary{shapes,decorations}
\usepackage{soul}
\usepackage{cancel}
\usepackage{natbib}

\newcommand{\mic}[1]{\textcolor{black}{#1}}

\newcommand{\pb}{STSWSN-PC}
\newcommand{\problem}{sub-tree scheduling for wireless sensor networks with partial coverage}

\newcommand{\auxiliary}{fictitious}

\newtheorem{lemma}{Lemma}
\newtheorem{proposition}{Proposition}
\newtheorem{corollary}{Corollary}
\newtheorem{example}{Example}
\newtheorem{observation}{Observation}

\usepackage[colorlinks=true]{hyperref}

\title{On the Complexity of Sub-Tree Scheduling for Wireless Sensor Networks with Partial Coverage}

\author{ \href{https://orcid.org/0000-0002-8521-896X}{\includegraphics[scale=0.06]{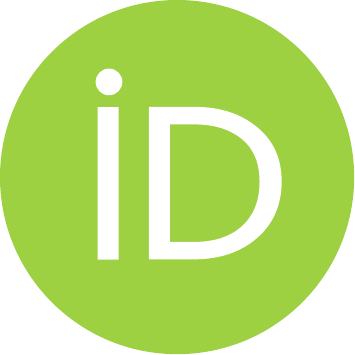}\hspace{1mm}Michele Barbato}\thanks{Corresponding author} \\
	Dipartimento di Informatica\\
	Università degli Studi di Milano \\
	via Celoria 18, 20133 Milano \\
	\texttt{michele.barbato@unimi.it} \\
	\And
	\href{https://orcid.org/0000-0002-5722-5476}{\includegraphics[scale=0.06]{orcid.pdf}\hspace{1mm}Nicola Bianchessi} \\
	Dipartimento di Informatica\\
	Università degli Studi di Milano \\
	via Celoria 18, 20133 Milano \\
	\texttt{nicola.bianchessi@unimi.it}
}

\date{}

\renewcommand{\headeright}{}
\renewcommand{\undertitle}{}
\renewcommand{\shorttitle}{On the complexity of STSWSN-PC}

\hypersetup{
pdftitle={On the Complexity of Sub-Tree Scheduling for Wireless Sensor Networks with Partial Coverage},
pdfauthor={Michele Barbato, Nicola Bianchessi},
pdfkeywords={Wireless sensor network, Sub-tree scheduling, Partial coverage, Complexity},
}

\begin{document}
\maketitle


\begin{abstract}
Given an undirected graph $G$ whose edge weights change over $s$ time slots, the \problem{} asks to partition the vertices of $G$ in $s$ non-empty trees such that the total weight of the trees is minimized. In this note we show that the problem is \textbf{NP}-hard in both the cases where $s$ $(i)$ is part of the input and 
$(ii)$ is a fixed instance parameter. In both our proofs we reduce from the cardinality Steiner tree problem. 
We additionally give polynomial-time algorithms for structured inputs of the problem.
\end{abstract}

\section{Introduction}
A central problem in the management of wireless sensor networks is to extend the lifetime of wireless sensors through operating policies ensuring energy efficiency and/or balancing.
Its importance stems from the fact that even a single failure of a wireless sensor can in principle compromise the effectiveness of the whole network.
From the viewpoint of energy balancing, a general \mic{approach} to minimize energy consumption is to split the set of sensors into several non-empty subsets and to subdivide the \mic{planning} horizon into as many slots, so that the subsets of sensors are operated sequentially, one at each time slot.

The \emph{\problem{}} (\pb{}), introduced by~\cite{adasme2019optimal}, is a particular implementation of such \mic{an approach}, with the additional requirements that the sensors operated simultaneously are mutually connected under a tree topology, and each sensor must be active in a unique time slot.
Namely, the \pb{} is defined on an undirected graph $G=(V,E)$ representing the network of sensors, a number $s$, $1\leq s \leq |V|$, of time slots, and vectors $w^1,w^2,\dots, w^s\in\mathbb R^E_+$ of edge-weights (one for each time slot). The aim is to find a set $T_1,T_2,\dots, T_s$ of non-empty vertex-disjoint trees of $G$ covering $V$ and minimizing $\sum_{i=1}^{s} w^i(T_i)$.
In the above description, the vertices of $G$ represent the sensors of the network, the edges represent direct links among
sensors, and the weights represent the time slot-dependent power for transmitting information over the corresponding edges.

The input of the \pb{} is simultaneously defined by the graph $G$, 
the number of time slots $s$, and 
the values of the edge weight vectors. The \pb{} may admit efficient optimization algorithms for structured inputs. For example, when  the weights are constant throughout the time slots (\emph{i.e.}, $w^i\equiv w^j$ for all $i,j=1,2,\dots,s$), the \pb{} is solvable in polynomial time, \emph{e.g.}, by using Kruskal's algorithm \citep{Kruskal56} and terminating it at the first iteration yielding a spanning forest with $s$ trees;
when $s=1$, the \pb{} boils down to the minimum spanning tree (MST) problem on general graphs and, as such, is solvable in polynomial time; 
when $s=|V|$, the optimal solution consists of \mic{arbitrarily} assigning one vertex to each time slot.

However, unstructured instances of the \pb{} have been tackled in~\cite{adasme2019optimal} and~\cite{bianchessi2022} by means of branch-and-bound and branch-and-cut algorithms, respectively. These approaches implicitly suggest that the problem is theoretically intractable, although its computational complexity is unknown to the best of our knowledge. The purpose of this note is to fill in this gap. 

In~\autoref{sec:unfixed} we study the complexity of the \pb{} when $s$ is part of the input, that is, $s$ is not fixed in $\{2,3,\dots, |V|-1\}$; in \autoref{sec:fixed} we study the complexity under the assumption that $s$ is an instance parameter with a prescribed value $\bar s\ge 2$. Through reductions from the (minimum weight) Steiner tree problem~\cite[p. 208]{GJ90}, we show that the \pb{} is \textbf{NP}-hard in both cases, thus justifying the usage of implicit enumeration schemes to solve it.
Finally, in~\autoref{sec:special} we discuss additional structured inputs, other than those mentioned above, for which the \pb{} is solvable in polynomial time.

\section{NP-hardness when the number of time slots is not fixed}\label{sec:unfixed}
Given an undirected connected graph $G=(V,E)$ with $|V|=n$ vertices and a subset $R\subset V$ of \emph{terminal} vertices, a \emph{Steiner tree} is a subtree $T$ of $G$ such that $R\subseteq V(T)$.
Given also a weight $w(e) \in \mathbb Z_+$ for each $e \in E$, computing the Steiner tree of minimum total edge-weight is in general \textbf{NP}-hard, and the problem remains \textbf{NP}-hard if all weights are equal \cite[p. 209]{GJ90}.
In particular, given $w(e) = 1$ for each $e \in E$, the pair $(G,R)$, and $k\in\mathbb Z_+$ with $|R|-1\leq k\leq n-2$, the \emph{cardinality Steiner tree} (CST) problem consisting of determining the existence of a Steiner tree of $G$ with at most $k$ edges is \textbf{NP}-complete.

We now show that an oracle solving the \pb{} in polynomial time allows to solve the CST in polynomial time, thus obtaining that the \pb{} is \textbf{NP}-hard.
We point out the CST with at most three terminals can be solved in polynomial time~\citep{arrighi_et_al}, therefore we restrict ourselves to CST instances with $|R|\ge 4$.

Given $k$ and $(G,R)$ defining a CST instance as above, we construct a new graph $\bar G=(\bar V,\bar E)$ obtained from $G$ by introducing $n-k$ \emph{\auxiliary\ vertices} $v_1^r,v_2^r,\dots,v_{n-k}^r$ for each terminal vertex $r\in R$ and defining $\bar E=E\cup E^R$, where $E^R=\{(v_j^r,r)\colon j=1,2,\dots,n-k,r\in R\}$; that is, each \auxiliary\ vertex is connected precisely to the corresponding terminal vertex.
An example of such a construction is given in~\autoref{fig:non-fixed}.

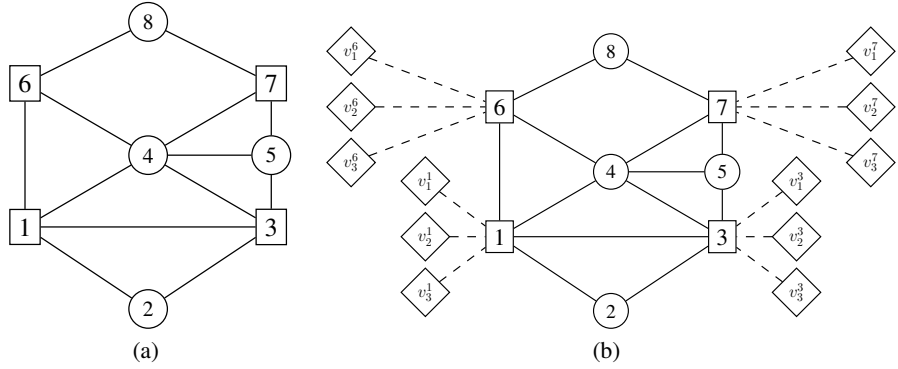
\begin{figure}
\centering

\subfigure[][]{
\resizebox{0.24\textwidth}{!}{
\begin{tikzpicture}[scale=.45, transform shape]
\tikzstyle{every node} = [circle,scale=1.5,fill=white,draw=black]

\node[rectangle,scale=1.2](t1) at (6,0.0) {1};
\node[rectangle,scale=1.2] (t3) at (12,0) {3};
\node[rectangle,scale=1.2] (t6) at (6,3.5) {6};
\node[rectangle,scale=1.2] (t7) at (12,3.5) {7};

\node (2) at (9.0,-2.0) {2};
\node (4) at (9,1.75) {4};
\node (5) at (12,1.75) {5};
\node (8) at (9,5.0) {8};

\draw (t1) -- (2);
\draw (t1) -- (t6);
\draw (t1) -- (4);
\draw (t3) -- (t1);
\draw (t3) -- (2);
\draw (t3) -- (4);
\draw (t3) -- (5);
\draw (4) -- (5);
\draw (4) -- (t6);
\draw (4) -- (t7);
\draw (5) -- (t7);
\draw (t6) -- (8);
\draw (t7) -- (8);

\end{tikzpicture}
}
}
\subfigure[][]{
\resizebox{0.47\textwidth}{!}{
\begin{tikzpicture}[scale=.45, transform shape]
\tikzstyle{every node} = [circle,scale=1.5,fill=white,draw=black]

\node[diamond,scale=0.6] (t101) at (4,1.5) {\Large $v^1_1$};
\node[diamond,scale=0.6] (t102) at (4,0) {\Large $v^1_2$};
\node[diamond,scale=0.6] (t103) at (4,-1.5) {\Large $v^1_3$};

\node[diamond,scale=0.6] (t301) at (14,1.5) {\Large $v^3_1$};
\node[diamond,scale=0.6] (t302) at (14,0) {\Large $v^3_2$};
\node[diamond,scale=0.6] (t303) at (14,-1.5) {\Large $v^3_3$};

\node[diamond,scale=0.6] (t601) at (2,5) {\Large $v^6_1$};
\node[diamond,scale=0.6] (t602) at (2,3.5) {\Large $v^6_2$};
\node[diamond,scale=0.6] (t603) at (2,2) {\Large $v^6_3$};

\node[diamond,scale=0.6] (t701) at (16,5) {\Large $v^7_1$};
\node[diamond,scale=0.6] (t702) at (16,3.5) {\Large $v^7_2$};
\node[diamond,scale=0.6] (t703) at (16,2) {\Large $v^7_3$};

\node[rectangle,scale=1.2](t1) at (6,0.0) {1};
\node[rectangle,scale=1.2] (t3) at (12,0) {3};
\node[rectangle,scale=1.2] (t6) at (6,3.5) {6};
\node[rectangle,scale=1.2] (t7) at (12,3.5) {7};

\node (2) at (9.0,-2.0) {2};
\node (4) at (9,1.75) {4};
\node (5) at (12,1.75) {5};
\node (8) at (9,5.0) {8};

\draw (t1) -- (2);
\draw (t1) -- (t6);
\draw (t1) -- (4);
\draw (t3) -- (t1);
\draw (t3) -- (2);
\draw (t3) -- (4);
\draw (t3) -- (5);
\draw (4) -- (5);
\draw (4) -- (t6);
\draw (4) -- (t7);
\draw (5) -- (t7);
\draw (t6) -- (8);
\draw (t7) -- (8);

\draw [dashed] (t701) -- (t7);
\draw [dashed] (t702) -- (t7);
\draw [dashed] (t703) -- (t7);

\draw [dashed] (t601) -- (t6);
\draw [dashed] (t602) -- (t6);
\draw [dashed] (t603) -- (t6);

\draw [dashed] (t301) -- (t3);
\draw [dashed] (t302) -- (t3);
\draw [dashed] (t303) -- (t3);

\draw [dashed] (t101) -- (t1);
\draw [dashed] (t102) -- (t1);
\draw [dashed] (t103) -- (t1);

\end{tikzpicture}
}
}
\caption{Example of a CST instance (a),  in which terminal vertices are squared-shaped, and of the corresponding \pb{} instance for $k = 5$ (b), in which fictitious vertices are diamond-shaped.}\label{fig:non-fixed}
\end{figure}

Next, we define a \pb{} instance $I$ on graph $\bar G$ and
$s = n-k+1$ time slots. 
Note that, since $|R|-1\le k\le n-2$ and $|R|\ge 4$, then $3\le s \le n-2<|\bar V|$, 
hence our definition of the number of time slots excludes the polynomially solvable cases of the \pb{}.

The weights of the time slots are defined as follows:

\begin{align}
w^1_e&=
\begin{cases}
0&\text{if }e\in E^R\\
1&\text{otherwise}
\end{cases}&\label{eqn:regular} \\
w^j_e&=\begin{cases}
n & \text{if }e\in E^R\\
1&\text{otherwise}
\end{cases}&\forall j=2,3,\dots,n-k+1\label{eqn:bigM}
\end{align}

\begin{lemma}\label{lemma:optimal_is_Steiner}
Let $T^\star_1,T^\star_2,\dots,T^\star_{n-k+1}$ be an optimal solution to $I$. The restriction of $T^\star_1$ to the vertices in $V$ is a Steiner tree of $G$.
\end{lemma}
\begin{proof}
Assume that the restriction of  $T^\star_1$ to the vertices of $G$ is not a Steiner tree. Then there is at least a terminal vertex $r^\star$ contained in a tree of a time slot after the first one. Since all trees $T^\star_1,T^\star_2,\dots,T^\star_{n-k+1}$ are connected,  at least one edge of $E^R$ belongs to that time slot. By~\eqref{eqn:bigM} the optimal solution to $I$ has value at least $n$.
Now we show the existence of a solution with better value. Namely, in the first time slot we consider the tree $T_1$ spanning all vertices of $\bar G$ except the $n-k$ \auxiliary\ vertices linked to $r^\star$ and we set $T_j=\{v^{r^\star}_{j-1}\}$ for $j=2,3,\dots,n-k+1$. Then $T_1,T_2,\dots,T_{n-k+1}$ is a feasible solution whose value is $n-1$ by~\eqref{eqn:regular}.
\end{proof}
Now we can prove the main result. In the proof, given $S\subseteq \bar V$, we denote by $\delta(S)$ its \emph{cut}, namely, the set of edges having one endpoint in $S$ and the other endpoint outside $S$. 
\begin{proposition}\label{prop:prop_1}
There exists a solution to the CST instance given by $k$ and $(G,R)$ if and only if the optimal solution to $I$ has value at most $k$. Therefore the \pb{} is \textbf{NP}-hard.
\end{proposition}
\begin{proof}
For the ``if'' part assume that there exists an optimal solution having value at most $k$; denoting by $T^\star$ the restriction of its tree of the first time slot to the vertices in $V$, the nonnegativity of the weights in~\eqref{eqn:regular} yields $|T^\star|\le k$. Then the result follows from~\autoref{lemma:optimal_is_Steiner}.

Now, let us prove the ``only if'' part. Assume that there
exists a Steiner tree $T$ of $G$ such that $|T|\le k$. We assume, without loss of generality, that $|T|=k$: otherwise we repeatedly update $T$ by adding one edge of $G$ belonging to $\delta(T)$, until reaching the required cardinality (this is always possible as $G$ is connected and since the update always returns a Steiner tree). Then, let $\bar v \in \bar V \setminus V$ be an arbitrary \auxiliary\ vertex, define $\hat{V} = \bar V \setminus \{V \cup \{\bar v\}\}$ as the set of remaining \auxiliary\ vertices, and let $ V^C=V \setminus V(T) = \{v_1,v_2,\dots,v_{n-k-1}\}$ be the vertices in the complement of $T$ in $G$ (as $|T| = k$, $T$ comprises $k + 1$ vertices).
We consider the feasible solution for $I$ given by $T_1=T\cup\delta(\hat{V})$, $T_2=\{\bar v\}$ and $T_j=\{v_{j-2} \in V^C\}$ for every $j=3,4,\dots,n-k+1$.
By~\eqref{eqn:regular}--\eqref{eqn:bigM} such a solution has value $k$. Then the optimal solution to $I$ has value at most $k$.
\end{proof}

In the above construction, $\bar G$ is obtained from $G$ by appending leaves to its terminal vertices. This is a minor modification of the initial graph, hence the \pb{} remains difficult on those classes of graphs which are closed under such modification and on which the CST is \textbf{NP}-complete.
It is the case of \emph{chordal bipartite graphs}, that is, bipartite graphs whose cycles $C$ of length at least 6 induce a subgraph with at least $|C|+1$ edges.
More precisely we have:
\begin{corollary}\label{cor:bip_chord_gr}
The \pb{} is \textbf{NP}-hard on bipartite chordal graphs.
\end{corollary}
\begin{proof}
Appending leaves to a subset of vertices of a bipartite chordal graph maintains the chordal bipartiteness.
Then the result follows from the \textbf{NP}-completeness of the CST on bipartite chordal graphs proved by~\cite{muller1987np}.
\end{proof}

\section{NP-hardness when the number of time slots is fixed}
\label{sec:fixed}
In this section we consider the complexity of the \pb{} by assuming that we have $s = \bar s$ time slots, with $\bar s\geq 2$ fixed, and we show that the problem remains \textbf{NP}-hard.

We modify the approach of previous section as follows. 
Let us consider a graph $G=(V,E)$ with $|V|=n$ vertices and a set $R \subset V$, $|R| \geq 4$, of terminal vertices defining an instance of the CST problem.
We define a graph $G^\star=(V^\star,E^\star)$ where $V^\star=V\cup V^R$, with $V^R=\{v^r_1,v^r_2,\dots,v^r_{\bar s-1}\colon r\in R\}$ being a set of \auxiliary\ vertices associated with those in $R$, and where $E^\star=E\cup E^C\cup E^R$, with $E^C=\{(v,w)\colon v,w\in V\text{ s.t. } (v,w)\not\in E\}$ and $E^R=\{(r,v^r_j)\colon r\in R,j=1,2,\dots,\bar s-1\}$. That is, $G^\star$ is obtained by extending $G$ to a complete graph and by linking each terminal vertex in $G$ to the corresponding $\bar s-1$ \auxiliary\ vertices (see~\autoref{fig:fixed} for an example). 

\begin{figure}
\centering

\subfigure[][]{
\resizebox{0.25\textwidth}{!}{
\begin{tikzpicture}[scale=.45, transform shape]
\tikzstyle{every node} = [circle,scale=1.5,fill=white,draw=black]

\node[rectangle,scale=1.2](t1) at (6,0.0) {1};
\node[rectangle,scale=1.2] (t3) at (12,0) {3};
\node[rectangle,scale=1.2] (t6) at (6,3.5) {6};
\node[rectangle,scale=1.2] (t7) at (12,3.5) {7};

\node (2) at (9.0,-2.0) {2};
\node (4) at (9,1.75) {4};
\node (5) at (12,1.75) {5};
\node (8) at (9,5.0) {8};

\draw (t1) -- (2);
\draw (t1) -- (t6);
\draw (t1) -- (4);
\draw (t3) -- (t1);
\draw (t3) -- (2);
\draw (t3) -- (4);
\draw (t3) -- (5);
\draw (4) -- (5);
\draw (4) -- (t6);
\draw (4) -- (t7);
\draw (5) -- (t7);
\draw (t6) -- (8);
\draw (t7) -- (8);
\end{tikzpicture}
}
}
\subfigure[][]{
\resizebox{0.4\textwidth}{!}{
\begin{tikzpicture}[scale=.45, transform shape]
\tikzstyle{every node} = [circle,scale=1.5,fill=white,draw=black]

\node[diamond,scale=0.6] (t101) at (4,0.75) {\Large $v^1_1$};
\node[diamond,scale=0.6] (t102) at (4,-0.75) {\Large $v^1_2$};

\node[diamond,scale=0.6] (t301) at (14,0.75) {\Large $v^3_1$};
\node[diamond,scale=0.6] (t302) at (14,-0.75) {\Large $v^3_2$};

\node[diamond,scale=0.6] (t601) at (4,4.25) {\Large $v^6_1$};
\node[diamond,scale=0.6] (t602) at (4,2.75) {\Large $v^6_2$};

\node[diamond,scale=0.6] (t701) at (14,4.25) {\Large $v^7_1$};
\node[diamond,scale=0.6] (t702) at (14,2.75) {\Large $v^7_2$};

\node[rectangle,scale=1.2](t1) at (6,0.0) {1};
\node[rectangle,scale=1.2] (t3) at (12,0) {3};
\node[rectangle,scale=1.2] (t6) at (6,3.5) {6};
\node[rectangle,scale=1.2] (t7) at (12,3.5) {7};

\node (2) at (9.0,-2.0) {2};
\node (4) at (9,1.75) {4};
\node (5) at (12,1.75) {5};
\node (8) at (9,5.0) {8};

\draw (t1) -- (2);
\draw (t1) -- (t6);
\draw (t1) -- (4);
\draw (t3) -- (t1);
\draw (t3) -- (2);
\draw (t3) -- (4);
\draw (t3) -- (5);
\draw (4) -- (5);
\draw (4) -- (t6);
\draw (4) -- (t7);
\draw (5) -- (t7);
\draw (t6) -- (8);
\draw (t7) -- (8);

\draw [dashed] (t601) -- (t6);
\draw [dashed] (t602) -- (t6);

\draw [dashed] (t701) -- (t7);
\draw [dashed] (t702) -- (t7);

\draw [dashed] (t301) -- (t3);
\draw [dashed] (t302) -- (t3);

\draw [dashed] (t101) -- (t1);
\draw [dashed] (t102) -- (t1);

\draw [densely dotted] (2) -- (4);
\draw [densely dotted] (2) -- (5);
\draw [densely dotted] (2) -- (t6);
\draw [densely dotted] (2) -- (t7);
\draw [densely dotted] (2) .. controls(7.5,2) .. (8);

\draw [densely dotted] (t3) .. controls(9,0.5) .. (t6);
\draw [densely dotted] (t3)  .. controls(13,1.75) .. (t7);
\draw [densely dotted] (t3) -- (8);

\draw [densely dotted] (4) -- (8);

\draw [densely dotted] (5) -- (t6);
\draw [densely dotted] (5) -- (8);
\draw [densely dotted] (5) -- (t1);

\draw [densely dotted] (t6) -- (t7);

\draw [densely dotted] (t7) .. controls(9,0.5) .. (t1);

\end{tikzpicture}
}
}
\caption{Example of a CST instance (a), in which terminal vertices are squared-shaped, and of the corresponding \pb{} instance for $\bar s = 3$ (b), in which fictitious vertices are diamond-shaped.}\label{fig:fixed}
\end{figure}
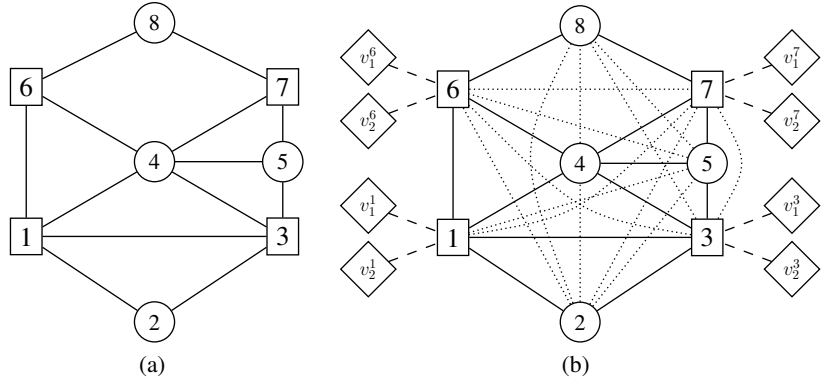

For every $e\in E^\star$ we define the following edge weights:
\begin{align}
w^1_e&=
\begin{cases}
0&\text{if }e\in E^R\\
1&\text{if }e\in E\\
n&\text{otherwise,}
\end{cases}&\label{eqn:ts1_fixed} \\
w^j_e&=
\begin{cases}
n&\text{if }e\in E^R\\
0&\text{otherwise.}
\end{cases}&\forall j=2,3,\dots,\bar s\label{eqn:ts2_fixed}
\end{align}

Let $I^\star$ be the resulting \pb{} instance. 

A Steiner tree $T$ of $G$ with $k$ edges corresponds to a solution $T_1,T_2,\dots,T_{\bar s}$ of $I^\star$ having value $k$. We distinguish two cases: 
\begin{enumerate}
    \item if $T$ is not spanning, let $r\in R$ be an arbitrary terminal vertex of $G$ and let $v_1^r,v_2^r,\dots,v_{\bar s-2}^r$ be $\bar s-2$ arbitrary fictitious vertices linked to $r$. One obtains $T_1$ by extending $T$ with all vertices in $V^R\setminus \{v_1^r,v_2^r,\dots,v_{\bar s-2}^r\}$ (whose linking edges in $E^R$  have weight 0 in the first time slot, by~\eqref{eqn:ts1_fixed}), by defining $T_2$ as the spanning tree of the complete graph $G^\star\setminus V(T_1)$ involving only edges in $E\cup E^C$ (which have weight 0 in the second time slot, by~\eqref{eqn:ts2_fixed}) and by defining $T_j=\{v_{j-2}^r\}$ for every $j=3,4,\dots\bar s$;\label{case:non-spanning}
    \item if $T$ is spanning, let $r\in R$ be an arbitrary terminal vertex of $G$ and let $v_1^r,v_2^r,\dots,v_{\bar s-1}^r$ be the $\bar s-1$ fictitious vertices linked to $r$. One obtains $T_1$ by extending $T$ with all vertices in $V^R\setminus\{v_1^r,v_2^r,\dots,v_{\bar s-1}^r\}$, and by defining $T_j=\{v^r_{j-1}\}$ for every $j=2,3,\dots,\bar s$.\label{case:spanning}
\end{enumerate}

Note that, since a spanning tree of $G$ is also a Steiner tree, the construction in the above case~\ref{case:spanning} shows that an optimal solution to $I^\star$ has value at most $n-1$.
Then, as in~\autoref{lemma:optimal_is_Steiner}~and ~\autoref{prop:prop_1}, it is possible to state that if $T^\star_1,T^\star_2,\dots,T^\star_{\bar s}$ is an optimal solution to $I^\star$, the restriction of $T^\star_1$ to the vertices in $V$ is a Steiner tree of $G$ having the same value.
Indeed, we first observe that $T^\star_1$ has its edges in $E\cup E^R$, as otherwise~\eqref{eqn:ts1_fixed} would imply that  the considered solution has weight at least $n$, contradicting its optimality; moreover, if $T^\star_1$ is not a Steiner tree of $G$, there should be a vertex $r\in R$ belonging to $T^\star_j$ with $2\le j\le \bar s$ and, since $T^\star_1,T^\star_2,\dots,T^\star_{\bar s}$ are connected, we have that at least one edge of $E^R$ is taken outside the first time slot; then by~\eqref{eqn:ts2_fixed}, the considered solution has value at least $n$, again contradicting its optimality.

The above arguments prove that the considered CST instance admits a solution if and only if the corresponding \pb{} instance has value at most $k$, hence we have:

\begin{proposition}\label{prop:nphard_fixed}
The \pb{} with a fixed number $\bar s\ge 2$ of time slots is \textbf{NP}-hard.
\end{proposition}
 
We remark that the transformation from $G$ to $G^\star$ used in the above reduction does not allow to state a result similar to~\autoref{cor:bip_chord_gr}.

\section{Structured polynomially-solvable cases}\label{sec:special}

The results of~\autoref{prop:prop_1} and~\autoref{prop:nphard_fixed} hold without making any assumption on the structure of the \pb{} instances.
Here we present two polynomially-solvable cases when the input is structured. The first one generalizes the approach described in the Introduction for the case $s=|V|$.

\begin{observation}
When $|V|-s$ is constant the \pb{} is solvable in polynomial time. 
\end{observation}
\begin{proof}
When $s=|V|-1$, a feasible solution contains one edge in a time slot and single vertices in all remaining time slots; then an optimal solution can be determined in $O(|V||E|)$ time by exhaustively listing all values $w^j_e$ for $e\in E$ and $1\le j\le |V|-1$ and considering the minimum one. A similar algorithm (of higher time complexity) can be exhibited for any constant value of $|V|-s$.
\end{proof}

The second polynomially-solvable case relates to the graph topology:
\begin{observation}\label{obs:trees}
If $G=(V,E)$ is a tree, the \pb{} with a fixed number $\bar s\ge 2$ of time slots is solvable in polynomial time.
\end{observation}

\begin{proof}
We can list in $O(n^{\bar s-1})$ all subsets of $\bar s-1$ edges whose removal decomposes $G$ into a forest with $\bar s$ trees.
For each such a subset we assign in polynomial time the corresponding trees to the $\bar s$ time slots solving a perfect matching on the weighted complete bipartite graph $\mathcal B=(\mathcal T;\mathcal S,\mathcal W)$ where each vertex in $\mathcal T$ represents a tree, each vertex of $\mathcal S$ represents a time slot and edge $e_{\tau\sigma}\in\mathcal W$ linking $\tau\in\mathcal T$ to $\sigma\in\mathcal S$ has weight $w^\sigma(\tau)$.
\end{proof}

\autoref{obs:trees} motivates the following questions that we leave open: $(i)$ When the number of time slots is not fixed, which is the complexity of the \pb{} defined on trees? $(ii)$ Are there any other graph families (other than trees) for which the \pb{} is solvable in polynomial time, at least when the number of time slots is fixed?

\section*{Acknowledgments}
The authors are grateful to Alberto Ceselli and to Emiliano Lancini for their comments on  the manuscript.

\bibliographystyle{natbib}
\bibliography{biblio} 

\end{document}